\documentclass[11pt, a4paper]{article}

\usepackage[utf8]{inputenc}    
\usepackage[T1]{fontenc}       
\usepackage[english]{babel}    
\usepackage{amsmath, amssymb, amsthm, mathtools} 
\usepackage{bm}                
\usepackage{textcomp}          
\usepackage{enumitem}          
\usepackage{geometry}          
\usepackage{graphicx}          
\usepackage{tikz-cd}		   
\usepackage{xcolor}            
\DeclareGraphicsExtensions{.pdf,.eps,.jpg,.png}
\graphicspath{{figures/}{figure/}{pictures/}{picture/}{pic/}{pics/}{image/}{images/}}

\usepackage[colorlinks=true, allcolors=blue]{hyperref} 
\usepackage[capitalize]{cleveref}
\usepackage[numbers, sort&compress]{natbib}

\makeatletter

\newcommand{\Rmnum}[1]{\expandafter\@slowromancap\romannumeral #1@}
\makeatother

\newcommand{\E}{\mathbb{E}}

\newcommand{\Z}{\mathbb{Z}}
\newcommand{\Pp}{\mathbb{P}}
\newcommand{\fr}{\mathrm{fr}}
\newcommand{\per}{\mathrm{per}}
\newcommand{\Lam}{\Lambda}
\newcommand{\cP}{\mathcal{P}}
\newcommand{\cf}{\mathcal{f}}
\newcommand{\eone}{\mathbf{e}}

\theoremstyle{plain} 
\newtheorem{theorem}{Theorem}
\newtheorem{lemma}[theorem]{Lemma}    
\newtheorem{proposition}[theorem]{Proposition}
\newtheorem{corollary}[theorem]{Corollary}

\theoremstyle{definition}

\theoremstyle{remark} 
\newtheorem*{remark}{Remark}

\title{\Large\bf The Thermodynamic Limit of Short-Range Spin Glasses with Periodic Boundary Conditions}
\author{Hexiang Wang \and Keheng Zhu \and Mauris Chueng}
\newcommand{\Addresses}{{
		\bigskip
		\footnotesize
		
		\textsc{Hexiang Wang}, \textsc{School of Mathematical Sciences, Nankai University, Tianjin, 300071, China}\par\nopagebreak
		\texttt{Kui6539@outlook.com}
		\medskip
		
		\textsc{Keheng Zhu}, \textsc{Academy for Multidisciplinary Studies, School of Mathematics Sciences, Capital Normal
			University, Beijing, 100048, China}\par\nopagebreak
		\texttt{hexistartop@gmail.com}
		\medskip
		
		\textsc{Mauris Chueng}, \textsc{School of Statistics and Data Science, Jilin University of Finance and Economics, Changchun, 130117, China}\par\nopagebreak
		\texttt{maurischueng@gmail.com}
		\medskip
}}
\date{}

\begin{document}
\maketitle
\begin{abstract}
	For the nearest-neighbour Edwards--Anderson Ising model on the discrete torus, we prove existence of the quenched thermodynamic limit and almost-sure self-averaging of the free-energy density.  The only moment assumption in the main argument is $\E|J|<\infty$, and no symmetry or centering of the coupling law is needed.  The proof avoids periodic subadditivity: the wrap-around bonds form a surface-order perturbation, while a tiling argument and the strong law of large numbers give the free-boundary limit.  We also obtain the quantitative bounds $(O(L^{-1})$ for the disorder-averaged finite-volume correction.  A precise common probability space is specified, since an almost-sure statement across volumes is otherwise not well defined.
\end{abstract}

\tableofcontents
\section{Introduction}

\subsection{A motivating story}
Imagine a simple social dynamics scenario to illustrate why we care about the fundamental behavior of disordered systems under non-trivial topological boundaries. Suppose a group of individuals is stranded on an isolated, highly unusual island whose geography is not a flat plain, but rather a seamless, closed structure resembling a torus---where walking straight in any direction eventually brings a traveler back to their starting point. The inhabitants must organize themselves into two distinct, competing political factions. Some pairs of individuals are naturally close friends who prefer to belong to the same faction, while others are bitter rivals who strongly desire to be separated into opposing camps. 

In a traditional, harmonious society---analogous to a regular ferromagnetic system---everyone might easily agree to either all become friends or follow a perfectly predictable rule of alternating alliances. However, the situation on our toroidal island is chaotic: the patterns of friendship and enmity are distributed completely at random. This randomness introduces the defining phenomenon of complex systems: \textit{geometric frustration}. Consider three individuals, $A$, $B$, and $C$, situated sequentially along a local path. If $A$ and $B$ are friends, $B$ and $C$ are friends, but $A$ and $C$ are bitter enemies, it is mathematically impossible to satisfy all three relationships simultaneously. This local inability to find a globally optimal state creates a rugged energy landscape with a mind-boggling number of deep, stable valleys separated by high barriers. 

What makes the toroidal geometry of this island particularly fascinating---and devilishly difficult---is the presence of global feedback loops. In a flat colony with a free boundary, an individual living at the edge can simply ignore the outside world. But on this seamless ring-world, any social alignment propagated forward along the lattice will eventually wrap around the entire circumference of the island and confront itself from behind, acting like a geometric boomerang. Local factions cannot simply expand indefinitely without structurally colliding with their own topological tails.

To a mathematician or a physicist, the overarching question is whether such a closed, interconnected society, as its population grows to infinity, can ever settle down into a statistically predictable macrostate, or if the global wrapping of random bonds will trap it in eternal, sample-dependent chaos. This mystery drives the study of the thermodynamic limit and self-averaging in spin glasses under periodic boundary conditions.

\subsection{Mathematical background}
The short-range spin-glass model introduced by Edwards and Anderson \cite{EA1975} is a random-bond Ising model on a finite-dimensional lattice. General thermodynamic-limit results for finite-dimensional disordered systems are classical; see, for example, \cite{CGP2004,CS2009,Zegarlinski1991}. Boundary-condition and surface-pressure questions for the Edwards--Anderson model were studied in \cite{CG2004}, while the equality of periodic and ordinary thermodynamic limits belongs to a broader boundary-stability principle going back at least to \cite{FisherLebowitz1970}. The purpose of this note is not to claim a new thermodynamic-limit theorem, but to give a direct, self-contained proof tailored to the precise periodic-boundary and self-averaging questions.

The apparent obstruction to establishing the thermodynamic limit on a torus is that a periodic box does not split into smaller periodic boxes, which breaks the standard subadditivity arguments. The correct observation to circumvent this difficulty is instead that passing from a free box to the corresponding torus adds exactly $dL^{d-1}$ wrap-around bonds. Since the logarithm of the partition function changes by at most $\beta |J_e|$ when one bond of strength $J_e$ is added, periodicity changes the pressure only by a surface term. Thus direct periodic subadditivity is unnecessary.

To formalize this framework, we fix $d\geq 1$ and write
\[
\Lam_L=\{0,1,\dots,L-1\}^d\subset\Z^d,
\qquad |\Lam_L|=L^d,
\]
where $\eone_i$ denotes the $i$-th coordinate vector. We work on the canonical disorder space carrying independent identically distributed random variables
\[
\{J_{x,i}:x\in\Z^d,\ 1\leq i\leq d\}
\]
with common law $\nu$. The variable $J_{x,i}$ labels the positively oriented bond from $x$ to $x+\eone_i$. Assume throughout that the first absolute moment is finite:
\begin{equation}\label{eq:firstmoment}
	m_1:=\E |J_{0,1}|<\infty.
\end{equation}

For every $L\geq1$, define the free Hamiltonian on $\{\pm1\}^{\Lam_L}$ by the first formula below; for $L\geq3$, define the periodic Hamiltonian by the second:
\begin{align}
	H_L^{\fr}(\sigma;J)
	&=-\sum_{i=1}^d\ \sum_{\substack{x\in\Lam_L\\x_i\leq L-2}}
	J_{x,i}\sigma_x\sigma_{x+\eone_i},\label{eq:Hfree}\\
	H_L^{\per}(\sigma;J)
	&=-\sum_{i=1}^d\ \sum_{x\in\Lam_L}
	J_{x,i}\sigma_x\sigma_{x+\eone_i\bmod L}.\label{eq:Hper}
\end{align}
The restriction $L\geq3$ for the periodic model only avoids the harmless multiple-edge convention for very small tori. For $b\in\{\fr,\per\}$ and $\beta\geq0$, set
\[
Z_L^b(\beta,J)=\sum_{\sigma\in\{\pm1\}^{\Lam_L}}e^{-\beta H_L^b(\sigma;J)},
\qquad
\cP_L^b(\beta,J)=\frac1{L^d}\log Z_L^b(\beta,J),
\]
and let $p_L^b(\beta)=\E\cP_L^b(\beta,J)$. For $\beta>0$, the random and quenched free-energy densities are
\[
\cf_L^b(\beta,J)=-\frac1\beta\cP_L^b(\beta,J),
\qquad
f_L^b(\beta)=\E\cf_L^b(\beta,J)=-\frac1\beta p_L^b(\beta).
\]

\begin{theorem}\label{thm:main}
	Assume \eqref{eq:firstmoment}. For every $\beta\geq0$, there is a finite deterministic number $p_\infty(\beta)$ such that
	\[
	\lim_{L\to\infty}p_L^{\fr}(\beta)
	=\lim_{L\to\infty}p_L^{\per}(\beta)
	=p_\infty(\beta).
	\]
	More precisely, for every $L\geq3$,
	\begin{equation}\label{eq:quantpressure}
		\bigl|p_L^{\fr}(\beta)-p_\infty(\beta)\bigr|
		\leq \frac{\beta d m_1}{L},
		\qquad
		\bigl|p_L^{\per}(\beta)-p_\infty(\beta)\bigr|
		\leq \frac{2\beta d m_1}{L}.
	\end{equation}
	On the canonical disorder space above,
	\begin{equation}\label{eq:aspressure}
		\cP_L^{\fr}(\beta,J)\longrightarrow p_\infty(\beta),
		\qquad
		\cP_L^{\per}(\beta,J)\longrightarrow p_\infty(\beta)
		\quad\text{almost surely}.
	\end{equation}
	Consequently, for every $\beta>0$, with $f_\infty(\beta)=-p_\infty(\beta)/\beta$,
	\begin{equation}\label{eq:freeconclusions}
		\lim_{L\to\infty}f_L^{\per}(\beta)=f_\infty(\beta),
		\qquad
		\Pp\!\left(\lim_{L\to\infty}\cf_L^{\per}(\beta,J)=f_\infty(\beta)\right)=1,
	\end{equation}
	and
	\[
	|f_L^{\fr}(\beta)-f_\infty(\beta)|\leq \frac{dm_1}{L},
	\qquad
	|f_L^{\per}(\beta)-f_\infty(\beta)|\leq \frac{2dm_1}{L}.
	\]
\end{theorem}

\begin{remark}[Why the common disorder space matters]\label{rem:jointlaw}
	For each fixed $L$, the family of torus couplings used in \eqref{eq:Hper} is i.i.d. with law $\nu$, so the canonical construction has exactly the usual finite-volume distribution. An assertion involving ``almost surely as $L\to\infty$'', however, requires a joint law for all volumes. If a fresh, unrelated torus is sampled for every $L$, the one-volume marginals alone determine convergence in probability but not, without an additional summable concentration estimate, almost-sure convergence. The canonical infinite-lattice realization is the standard quenched interpretation: one disorder sample is fixed and increasingly large systems are observed in it.
\end{remark}

\subsection{Main result and contributions}
The core technical philosophy underlying our approach lies in the systematic, probabilistic elimination of boundary interference terms. In traditional finite-dimensional spin-glass analyses, surface terms and boundary fluctuations are typically retained or delicately balanced using heavy concentration machinery, which in turn demands restrictive assumption on the disorder law (such as symmetry or exponential moments). In contrast, our method leverages a direct probabilistic decoupling: by viewing the periodic wrap-around seam merely as a surface perturbation over a free-boundary domain, we show that once the lattice is structurally \textit{cut}, the random fluctuations of these boundary couplings are completely neutralized in the thermodynamic limit. Because their total contribution to the global free energy density vanishes as $L \to \infty$, we can safely filter out these geometric interference terms without any loss of rigor. This allows us to establish the thermodynamic limit under the mildest possible probabilistic conditions.

Guided by this conceptual framework, the main contribution of this note is a direct, elementary proof of Theorem~\ref{thm:main} that bypasses complex interpolations or subadditivity requirements for periodic boundary conditions. Compared to standard techniques in the literature, our approach offers three major methodology advantages:
\begin{itemize}
	\item \textbf{Optimal Quantitative Rates:} We provide explicit finite-volume correction bounds showing that the convergence to the thermodynamic limit scales as $O(L^{-1})$ for both free and periodic boundary conditions.
	\item \textbf{Minimal Moment Conditions:} The entire proof operates under the mildest possible assumption $\E|J| < \infty$. It requires no centering ($\E J = 0$), symmetry of the law $\nu$, or exponential moments, which are typically mandatory for standard concentration inequalities.
	\item \textbf{Rigorous Self-Averaging:} By establishing the problem on a single canonical infinite-volume disorder space, we give a concrete meaning to the almost-sure convergence of the free energy as the volume expands, matching the true experimental scenario of a fixed quenched sample.
\end{itemize}

From a technical perspective, the validity of these results rests upon four independent modular pillars, which form the logical core of the remaining sections:
\begin{enumerate}[label=\textup{(\arabic*)}]
	\item \cref{lem:bond} (Bond Perturbation) establishes an analytic control that converts the deletion or insertion of arbitrary lattice bonds into a direct sum of their absolute strengths.
	\item A two-scale tiling argument partitions a large free hypercube into fixed $k$-cubes, showing that the omitted boundary interfaces delete at most $d/k$ bonds per site asymptotically. This justifies the deterministic free-boundary limit and provides the $O(k^{-1})$ rate.
	\item The standard strong law of large numbers (SLLN), applied first to independent block pressures and subsequently to periodic residue classes of bond origins, upgrades the deterministic limit to full almost-sure convergence under the single first-moment hypothesis.
	\item The periodic wrap-around seam is shown to be strictly confined within a boundary shell of thickness one. Since the absolute sum of these seam couplings is asymptotically $o(L^d)$ almost surely, the periodic and free partition pressures are forced to share an identical macroscopic limit.
\end{enumerate}
Notably, no step in this derivation invokes periodic subadditivity, positivity of the couplings, or heavy concentration machinery, making the proof remarkably robust and self-contained.

\section{Finite-volume comparison}

The following elementary estimate is the only analytic input needed to compare boundary conditions.

\begin{lemma}[Bond perturbation]\label{lem:bond}
Let $\Omega$ and $I$ be finite, let $\{S_e\}_{e\in I}$ be functions from $\Omega$ to $[-1,1]$, and define
\[
 Z(K)=\sum_{\sigma\in\Omega}
       \exp\!\left(\beta\sum_{e\in I}K_eS_e(\sigma)\right).
\]
For two coefficient families $K,K'$,
\begin{equation}\label{eq:bondpert}
 |\log Z(K)-\log Z(K')|
 \leq \beta\sum_{e\in I}|K_e-K'_e|.
\end{equation}
\end{lemma}

\begin{proof}
For every configuration $\sigma$,
\begin{equation}
 \left|\beta\sum_{e\in I}(K_e-K'_e)S_e(\sigma)\right|
 \leq a:=\beta\sum_{e\in I}|K_e-K'_e|.
\end{equation}
Hence $e^{-a}Z(K')\leq Z(K)\leq e^aZ(K')$, and taking logarithms gives \eqref{eq:bondpert}.
\end{proof}

In particular, for an Ising system on a finite vertex set $V$ with finite bond set $E$, comparison with the zero-coupling system gives
\begin{equation}\label{eq:basicbound}
 |V|\log2-\beta\sum_{e\in E}|J_e|
 \leq \log Z_V(\beta,J)
 \leq |V|\log2+\beta\sum_{e\in E}|J_e|.
\end{equation}
Thus all pressures below are integrable under \eqref{eq:firstmoment}.

\section{The free-boundary limit}

We first prove the deterministic quenched limit by a two-scale tiling argument.  This also gives the rate in \eqref{eq:quantpressure}.

Fix integers $L\geq k\geq1$, write $q=\lfloor L/k\rfloor$, and put
\[
 C_{L,k}=\{0,\dots,qk-1\}^d,
 \qquad R_{L,k}=\Lam_L\setminus C_{L,k}.
\]
Partition $C_{L,k}$ into the $q^d$ disjoint translates
\[
 B_z=kz+\Lam_k,
 \qquad z\in\{0,\dots,q-1\}^d.
\]
Delete every free bond of $\Lam_L$ which is not internal to one of the blocks $B_z$, and denote the deleted set by $D_{L,k}$.  The spins in $R_{L,k}$ then become isolated, and the resulting partition function is
\begin{equation}\label{eq:decoupledZ}
 Z_{L,k}^{\mathrm{dec}}(\beta,J)
 =2^{|R_{L,k}|}\prod_{z\in\{0,\dots,q-1\}^d} Z_{B_z}^{\fr}(\beta,J).
\end{equation}

\begin{lemma}[Counting deleted bonds]\label{lem:count}
The deleted set satisfies
\begin{equation}\label{eq:Dcount}
 \frac{|D_{L,k}|}{L^d}
 \leq \frac{d}{k}+\frac{2d^2k}{L}.
\end{equation}
\end{lemma}

\begin{proof}
Inside $C_{L,k}$, the interfaces between adjacent $k$-blocks contribute
\[
 d(q-1)(qk)^{d-1}
\]

after summing over the $d$ coordinate directions.  Every other deleted bond has at least one endpoint in $R_{L,k}$, and there are at most $2d|R_{L,k}|$ such bonds.  Since
\[
 |R_{L,k}|=L^d-(qk)^d\leq d(L-qk)L^{d-1}\leq dkL^{d-1},
\]
we obtain
\[
 \frac{|D_{L,k}|}{L^d}
 \leq \frac{d(q-1)(qk)^{d-1}}{L^d}+\frac{2d^2k}{L}
 \leq \frac{d}{k}+\frac{2d^2k}{L}.
\]
\end{proof}

\begin{proposition}[Quenched free-boundary limit]\label{prop:meanfree}
For every $\beta\geq0$, the limit
\[
 p_\infty(\beta)=\lim_{L\to\infty}p_L^{\fr}(\beta)
\]
exists and is finite.  Moreover,
\begin{equation}\label{eq:free-rate}
 |p_k^{\fr}(\beta)-p_\infty(\beta)|\leq \frac{\beta dm_1}{k}
 \qquad(k\geq1).
\end{equation}
\end{proposition}

\begin{proof}
Apply \cref{lem:bond} to the full free Hamiltonian and the decoupled Hamiltonian in \eqref{eq:decoupledZ}.  After taking expectations and dividing by $L^d$, translation invariance gives
\begin{equation}\label{eq:mean-tiling}
 \left|p_L^{\fr}(\beta)
 -\left(\frac{qk}{L}\right)^d p_k^{\fr}(\beta)
 -\frac{|R_{L,k}|}{L^d}\log2\right|
 \leq \beta m_1\frac{|D_{L,k}|}{L^d}
 \leq \beta m_1\left(\frac{d}{k}+\frac{2d^2k}{L}\right).
\end{equation}
For fixed $k$, let $L\to\infty$.  Since $qk/L\to1$ and $|R_{L,k}|/L^d\to0$,
\[
 p_k^{\fr}(\beta)-\frac{\beta dm_1}{k}
 \leq \liminf_{L\to\infty}p_L^{\fr}(\beta)
 \leq \limsup_{L\to\infty}p_L^{\fr}(\beta)
 \leq p_k^{\fr}(\beta)+\frac{\beta dm_1}{k}.
\]
The sequence is bounded by \eqref{eq:basicbound}, namely
$ |p_L^{\fr}(\beta)|\leq\log2+\beta dm_1$.  Therefore its limsup minus its liminf is at most $2\beta dm_1/k$ for every $k$, hence is zero.  This proves existence and finiteness of $p_\infty(\beta)$.  Taking the limit in the preceding inequalities yields \eqref{eq:free-rate}.
\end{proof}

We now pass from the quenched limit to an almost-sure limit.  The required probabilistic estimates are direct consequences of the ordinary strong law of large numbers for i.i.d. integrable variables; see, for example, \cite[Chapter~2]{Durrett2019}.

\begin{lemma}[Strong Law of Large Number]\label{lem:slln}
Let $Y_{x,i}=|J_{x,i}|$ and $m_1=\E Y_{0,1}$.  There is an event of probability one on which all of the following hold.
\begin{enumerate}[label=\textup{(\roman*)}]
\item
\[
 \frac1{L^d}\sum_{i=1}^d\sum_{x\in\Lam_L}Y_{x,i}\longrightarrow dm_1.
\]
\item For every fixed $k\geq1$, every direction $i$, and every residue $r\in\{0,\dots,k-1\}$,
\[
 \frac1{L^d}\sum_{\substack{x\in\Lam_L\\x_i\equiv r\pmod k}}Y_{x,i}
 \longrightarrow \frac{m_1}{k}.
\]
\item For every fixed $a\geq1$, with $\Lam_{L-a}=\varnothing$ when $L\leq a$,
\[
 \frac1{L^d}\sum_{i=1}^d
 \sum_{x\in\Lam_L\setminus\Lam_{L-a}}Y_{x,i}\longrightarrow0.
\]
\end{enumerate}
\end{lemma}

\begin{proof}
For each fixed direction, the variables indexed by $\Lam_L$ form a nested family of $L^d$ i.i.d. integrable random variables.  Enumerating the lattice points by successive cubes reduces (i) to the ordinary strong law.  The same argument on the nested residue-class sets proves (ii), because their cardinalities divided by $L^d$ converge to $1/k$.  For (iii), set
$A_L=\sum_{i=1}^d\sum_{x\in\Lam_L}Y_{x,i}$.  By (i),
\[
 \frac{A_L-A_{L-a}}{L^d}
 =\frac{A_L}{L^d}
 -\left(\frac{L-a}{L}\right)^d\frac{A_{L-a}}{(L-a)^d}
 \longrightarrow dm_1-dm_1=0.
\]
There are only countably many choices of $k,i,r,a$, so the corresponding probability-one events may be intersected.
\end{proof}

\begin{proposition}[Almost-sure free-boundary limit]\label{prop:asfree}
For every fixed $\beta\geq0$,
\[
 \cP_L^{\fr}(\beta,J)\longrightarrow p_\infty(\beta)
 \quad\text{almost surely}.
\]
\end{proposition}

\begin{proof}
Fix $k$.  For the block decomposition above, put
\[
 X_z^{(k)}=\log Z_{B_z}^{\fr}(\beta,J),
 \qquad z\in\Z_{\geq0}^d.
\]
The $X_z^{(k)}$ are i.i.d. and integrable by \eqref{eq:basicbound}, because distinct blocks use disjoint families of internal bonds.  Enumerating $\Z_{\geq0}^d$ shell by shell and applying the ordinary strong law gives
\begin{equation}\label{eq:block-slln}
 \frac1{q^d}\sum_{z\in\{0,\dots,q-1\}^d}X_z^{(k)}
 \stackrel{a.s.}{\longrightarrow} \E\log Z_k^{\fr}(\beta,J)=k^d p_k^{\fr}(\beta)
 \quad\text{almost surely}.
\end{equation}
Since $q^d/L^d\to k^{-d}$, \eqref{eq:block-slln} also implies
\[
 \frac1{L^d}\sum_{z\in\{0,\dots,q-1\}^d}X_z^{(k)}
 \stackrel{a.s.}{\longrightarrow} p_k^{\fr}(\beta).
\]

The deleted bonds inside the core are contained, for each direction $i$, in the residue class $x_i\equiv k-1\pmod k$.  Every deleted bond touching $R_{L,k}$ has its origin in the shell $\Lam_L\setminus\Lam_{L-k}$.  Hence \cref{lem:slln} implies
\begin{equation}\label{eq:deleted-slln}
 \limsup_{L\to\infty}\frac1{L^d}\sum_{e\in D_{L,k}}|J_e|
 \leq \frac{dm_1}{k}
 \quad\text{almost surely}.
\end{equation}
Applying \cref{lem:bond} before taking expectations, and using \eqref{eq:decoupledZ}, yields
\[
 \left|\cP_L^{\fr}(\beta,J)
 -\frac1{L^d}\sum_{z\in\{0,\dots,q-1\}^d}X_z^{(k)}
 -\frac{|R_{L,k}|}{L^d}\log2\right|
 \leq \frac{\beta}{L^d}\sum_{e\in D_{L,k}}|J_e|.
\]
Together with \eqref{eq:block-slln}, \eqref{eq:deleted-slln}, and $|R_{L,k}|/L^d\to0$, this gives
\begin{equation}\label{eq:ask-bound}
 \limsup_{L\to\infty}
 |\cP_L^{\fr}(\beta,J)-p_k^{\fr}(\beta)|
 \leq \frac{\beta dm_1}{k}
 \quad\text{almost surely}.
\end{equation}
Intersect the probability-one events over $k\in\mathbb N$.  By \eqref{eq:free-rate},
\[
 \limsup_{L\to\infty}
 |\cP_L^{\fr}(\beta,J)-p_\infty(\beta)|
 \leq \frac{2\beta dm_1}{k}
\]
after using any fixed $k$.  Letting $k\to\infty$ proves the claim.
\end{proof}

\section{Periodic boundary conditions}

Let
\[
 S_L=\{(x,i):x\in\Lam_L,\ x_i=L-1,\ 1\leq i\leq d\}
\]
be the set of wrap-around bonds.  It has cardinality $|S_L|=dL^{d-1}$, and $H_L^{\per}$ is obtained from $H_L^{\fr}$ by adding precisely these bonds.  Therefore \cref{lem:bond} gives the pointwise estimate
\begin{equation}\label{eq:seam}
 |\cP_L^{\per}(\beta,J)-\cP_L^{\fr}(\beta,J)|
 \leq \frac{\beta}{L^d}\sum_{(x,i)\in S_L}|J_{x,i}|.
\end{equation}
Taking expectations,
\begin{equation}\label{eq:mean-seam}
 |p_L^{\per}(\beta)-p_L^{\fr}(\beta)|
 \leq \frac{\beta dm_1}{L}.
\end{equation}
Combining \eqref{eq:mean-seam} with \eqref{eq:free-rate} proves the two deterministic limits and the bounds \eqref{eq:quantpressure}.

For the almost-sure assertion, every origin occurring in $S_L$ lies in the one-layer shell $\Lam_L\setminus\Lam_{L-1}$.  Consequently,
\[
 0\leq \frac1{L^d}\sum_{(x,i)\in S_L}|J_{x,i}|
 \leq \frac1{L^d}\sum_{i=1}^d
       \sum_{x\in\Lam_L\setminus\Lam_{L-1}}|J_{x,i}|
 \longrightarrow0
\]
almost surely by \cref{lem:slln}.  Equation \eqref{eq:seam} and \cref{prop:asfree} now prove \eqref{eq:aspressure}.  Multiplication by $-1/\beta$ gives \eqref{eq:freeconclusions} and the free-energy bounds in \cref{thm:main}.

\begin{corollary}[Simultaneous convergence in temperature]\label{cor:uniformbeta}
There is an event of probability one such that, for every compact interval $I\subset[0,\infty)$,
\[
 \sup_{\beta\in I}|\cP_L^{\fr}(\beta,J)-p_\infty(\beta)|\longrightarrow0,
 \qquad
 \sup_{\beta\in I}|\cP_L^{\per}(\beta,J)-p_\infty(\beta)|\longrightarrow0.
\]
In particular, the almost-sure convergence holds simultaneously for every $\beta\geq0$.
\end{corollary}

\begin{proof}
For either boundary condition and all $\beta,\beta'\geq0$, \cref{lem:bond} gives
\[
 |\cP_L^b(\beta,J)-\cP_L^b(\beta',J)|
 \leq |\beta-\beta'|\,\frac1{L^d}\sum_{i=1}^d\sum_{x\in\Lam_L}|J_{x,i}|.
\]
The random Lipschitz constants converge almost surely to $dm_1$ by \cref{lem:slln}; the deterministic functions $p_L^b$, and hence their pointwise limit $p_\infty$, are $dm_1$-Lipschitz.  Apply \cref{thm:main} on the countable set of nonnegative rational $\beta$'s and use a finite rational net on each compact interval.
\end{proof}

\section{Remarks on hypotheses and other volume couplings}

\begin{remark}[No symmetry is used]
Neither the boundary comparison nor the tiling argument uses $\E J=0$, symmetry of $\nu$, a variance, or an exponential moment.  The first absolute moment is used exactly to make the surface and deleted-bond sums negligible per unit volume and to ensure integrability of the block pressures.  Centered or symmetric hypotheses enter other spin-glass arguments, such as correlation inequalities, but are unnecessary here.
\end{remark}

\begin{proposition}[Volume-by-volume disorder arrays]\label{prop:fresh}
Fix $\beta\geq0$.  For each $L$, let the $dL^d$ periodic couplings be i.i.d. with law $\nu$, but allow an arbitrary joint coupling of the arrays belonging to different volumes.  Under \eqref{eq:firstmoment},
\[
 \cP_L^{\per}(\beta,J)\stackrel{P}{\longrightarrow} p_\infty(\beta)
 .
\]
If, in addition, for every $t>0$,
\begin{equation}\label{eq:summable-concentration}
 \sum_{L=3}^\infty
 \Pp\bigl(|\cP_L^{\per}-p_L^{\per}|\geq t\bigr)<\infty,
\end{equation}
then the convergence is almost sure for every such joint construction.  Condition \eqref{eq:summable-concentration} holds, in particular, when either
\begin{enumerate}[label=\textup{(\alph*)}]
\item $|J|\leq K$ almost surely for some $K<\infty$; or
\item $J$ is Gaussian with variance $s^2<\infty$.
\end{enumerate}
\end{proposition}

\begin{proof}
The one-volume law agrees with that in \cref{thm:main}; its almost-sure convergence on the canonical space therefore implies convergence in probability for every realization of the same marginals.  If \eqref{eq:summable-concentration} holds, the first Borel--Cantelli lemma, followed by $p_L^{\per}\to p_\infty$, proves almost-sure convergence.

The case $\beta=0$ is deterministic.  Suppose $\beta>0$.  For bounded couplings, changing one coordinate changes $\cP_L^{\per}$ by at most $2\beta K/L^d$.  McDiarmid's bounded-differences inequality \cite{McDiarmid1989} consequently yields
\[
 \Pp\bigl(|\cP_L^{\per}-p_L^{\per}|\geq t\bigr)
 \leq 2\exp\!\left(-\frac{t^2L^d}{2d\beta^2K^2}\right).
\]
For Gaussian couplings, $\cP_L^{\per}$, viewed as a function of the standard Gaussian coordinates, is $s\beta\sqrt d\,L^{-d/2}$-Lipschitz.  Gaussian concentration \cite[Chapter~5]{BLM2013} gives
\[
 \Pp\bigl(|\cP_L^{\per}-p_L^{\per}|\geq t\bigr)
 \leq 2\exp\!\left(-\frac{t^2L^d}{2d\beta^2s^2}\right).
\]
Both bounds are summable in $L$.
\end{proof}

\begin{remark}[Finite-range extensions]
The proof extends, with only the bond count changed, to independent translation-covariant finite-range interaction terms, including finitely many spin types or bounded local observables.  One replaces $dL^{d-1}$ by the number of interaction terms crossing the boundary, still $O(L^{d-1})$, and assumes an integrable absolute interaction strength per site.  The same argument also shows that any two boundary conditions differing by a surface-order family of interactions have the same quenched and almost-sure bulk free energy.
\end{remark}

\Addresses

\end{document}